\documentclass[acmsmall,nonacm]{acmart}

\AtBeginDocument{%
  }

\usepackage{bm}
\usepackage{xspace}
\usepackage{algorithmicx}
\usepackage{algorithm}
\usepackage{algpseudocode}
\usepackage[colorinlistoftodos]{todonotes}

\usepackage{tikz}
\usetikzlibrary{calc}
\usetikzlibrary{shapes,decorations.markings,arrows.meta,fit}

\newtheorem{thm}{Theorem}[section]

\newtheorem{prop}[thm]{Proposition}

\theoremstyle{definition}              

\newtheorem{example}[thm]{Example}

\newtheorem{note}[thm]{Note}

\newcommand{\R}{\mathcal{R}}
\newcommand{\C}{\mathcal{C}}
\renewcommand{\C}{\mathcal{C}}

\newcommand{\reld}{\mathbf{D}}

\newcommand{\res}{\upharpoonright}
\renewcommand{\deg}{\mathrm{deg}}
\newcommand{\maxdeg}{\mathrm{maxdeg}}
\newcommand{\avgdeg}{\mathrm{avgdeg}}
\newcommand{\ans}{\mathrm{Answers}}
\newcommand{\polylog}{\mathrm{polylog}}
\newcommand{\Reals}{\mathbb{R}}

\newcommand{\subw}{\textsf{subw}}
\newcommand{\out}{\mathsf{OUT}}

\newcommand{\para}[1]{\medskip \noindent \textbf{#1.}}

\newcommand{\cleanup}{\mathrm{Cleanup}}
\newcommand{\guanaco}{\mathrm{Guanaco}}
\newcommand{\guanacorec}{\mathrm{Guanaco\textup{-}Recurse}}

\newcommand{\rev}[1]{#1}

\begin{document}

\title{Guanaco: A Global-Uniformity Algorithm for Near-Submodular-Width Conjunctive Query Evaluation}

\author{Mahmoud Abo Khamis}
\orcid{0000-0003-3894-6494}
\email{mahmoud.abokhamis@cs.ox.ac.uk}
\affiliation{
  \institution{University of Oxford}
  \city{Oxford}
  \country{UK}
}

\author{Hubie Chen}
\orcid{0009-0005-4025-8086}
\email{hubie.chen@kcl.ac.uk}
\affiliation{
  \institution{King's College London}
  \city{London}
  \country{UK}
}

\begin{abstract}
We present Guanaco, an algorithm for performing conjunctive query evaluation where, for each Boolean conjunctive query, and positive epsilon, the algorithm achieves polynomial time with exponent equal to the submodular width plus epsilon.  The algorithm and its running time generalize smoothly to general conjunctive queries.  We believe the algorithm and its analysis to be notably simple, indeed, together we believe they form a highly simple argument that conjunctive query evaluation can be performed in essentially submodular width time.  In the case of Boolean conjunctive queries, the algorithm is based on interleaving three simple primitives: a subroutine for establishing a form of consistency; a subroutine for establishing global uniformity, which, briefly speaking, partitions relations as needed to control discrepancies between average degree and maximum degree; and, a simple step that joins pairs of existing relations to form new relations.
\end{abstract}



\maketitle

\section{Introduction}
\label{sec:intro}

\para{Background}
\emph{Conjunctive queries} are frequently occurring and heavily studied database queries~\cite{AbiteboulHullVianu95-foundationsdatabases,SchweikardtSchwentickSegoufin09-querylanguages}, and the computational problem of evaluating a conjunctive query on a relational database---known as \emph{conjunctive query evaluation}---is
intensely studied in databases and in computer science at large~\cite{FlumGrohe06-parameterizedcomplexity,SchweikardtSchwentickSegoufin09-querylanguages}.
Indeed, this problem can be viewed as a formulation of the \emph{constraint satisfaction problem}~\cite{KolaitisVardi00-containment,chen2023computability}, and overall has been studied from a multitude of motivations and perspectives~\cite{AbiteboulHullVianu95-foundationsdatabases,GottlobLeoneScarcello00-comparison,GroheSchwentickSegoufin01-conjunctivequeries,Grohe07-otherside,SchweikardtSchwentickSegoufin09-querylanguages,Marx10-canyoubeat,Marx11-csptruthtables,DBLP:journals/jacm/Marx13,GroheMarx14-fractional-edge-covers,DBLP:conf/pods/GottlobGLS16,DBLP:conf/pods/KhamisNR16,ChenMueller17-hierarchy,GottlobLPR21-analysis-decompositions,chen2023computability}.

It is now well-known that for each conjunctive query, there exists a polynomial-time algorithm for query evaluation, that is, there exists a value $\alpha \geq 1$ and a query evaluation algorithm with running time $O(N^{\alpha})$, where $N$ denotes the input database size.  Given the centrality of this evaluation problem, a number of works have sought to obtain a fine-grained understanding of the minimum exponent~$\alpha$ achievable for each conjunctive query.  While there is not yet a comprehensive understanding thereof, a number of recent works have shown general upper bounds in terms of the
\emph{submodular width}, a hypergraph measure that was introduced by Marx~\cite{DBLP:journals/jacm/Marx13} and is based on the notion of submodularity.  We focus our discussion on Boolean conjunctive queries, although the algorithms we discuss can operate on general conjunctive queries.  We use $H$ to denote the hypergraph underlying conjunctive queries under discussion.
\begin{itemize}

\item Building on an algorithm and analysis in Marx's original article~\cite{DBLP:journals/jacm/Marx13} introducing submodular width, Berkholz and Schweikardt~\cite{berkholz_et_al:LIPIcs.MFCS.2019.58} presented an algorithm, which (following~\cite{jaguar}) we call the \emph{Marx-Berkholz-Schweikardt} algorithm.  This algorithm achieves, for each conjunctive query and $\epsilon > 0$, a runtime of
$O(N^{2 \cdot \subw(H)+\epsilon})$.  Main ingredients used by this algorithm include two notions of consistency and a general procedure that partitions relations to enforce \emph{uniformity}, which (briefly) means that the discrepancy between average degree and maximum degree, in relations, is controlled.

\item Abo Khamis, Ngo, and Suciu~\cite{DBLP:conf/pods/Khamis0S17,theoretics:13722} presented an information-theoretic algorithm called \emph{PANDA} that achieves a runtime of $O(N^{\subw(Q)} \cdot \polylog(N))$.
Subsequently, they introduced an improved version of the algorithm, called \emph{PANDAExpress}~\cite{panda-express} that mitigates the $\polylog(N)$ factor to leave only a single $\log (N)$,
that is, that runs in time $O(N^{\subw(Q)} \cdot \log(N))$.

\item Abo Khamis and Chen~\cite{jaguar} presented an algorithm called \emph{Jaguar} achieving, for each $\epsilon~>~0$, a runtime of
$O(N^{\subw(H)+\epsilon})$.  This algorithm, throughout an execution, maintains a set function defined on all variable subsets and reasons about submodularity to guide a process of computing joins and partitioning relations.
\end{itemize}
On a high level, these algorithms collectively make use of similar computational primitives, namely, consistency/semijoin reduction,
computing joins to introduce new relations, and
data partitioning.  This consensus suggests that a submodular-width runtime is a natural barrier for algorithms employing these primitives.  Related to this, Fan, Koutris, and Zhao~\cite{fan_et_al:LIPIcs.ICALP.2023.127} presented techniques for lower bounding the runtime of conjunctive query evaluation, and showed that, for a number of different conjunctive queries, no combinatorial algorithm can achieve runtime
$O(N^{\subw(H)-\epsilon})$, assuming established fine-grained complexity-theoretic hypotheses.

\para{Contributions}
In this article, we build on the notions, algorithms, and analyses of Marx~\cite{DBLP:journals/jacm/Marx13} and Berkholz and Schweikardt~\cite{berkholz_et_al:LIPIcs.MFCS.2019.58}
to give a new algorithm for conjunctive query evaluation
that achieves, for each $\epsilon > 0$, a runtime
of $O(N^{\subw(H)+\epsilon})$.
We believe our algorithm, called \emph{Guanaco}, and
its analysis to be remarkably simple, and to
show a novel way in which the basic ingredients used by
the above algorithms can be plainly combined to achieve
conjunctive query evaluation in essentially submodular-width time.
Indeed, we believe that the algorithm and analysis constitute one of the simplest explanations for showing the claimed time bound on conjunctive query evaluation, and could be used for teaching this query evaluation result.

In the Boolean case, the algorithm is based on
interleaving three basic and generic primitives:
a subroutine for establishing a form of consistency;
 a subroutine for establishing \emph{global uniformity},
which reduces skew between average degree and maximum degree;
and a simple step that joins pairs of existing relations to form new relations.  In the general case, a version of the well-known Yannakakis algorithm for finding answers to acyclic conjunctive queries is also employed.
\rev{In contrast to the PANDA family of algorithms and the Jaguar algorithm,
the execution of Guanaco is not guided explicitly by submodularity or polymatroids.  In terms of its use of submodularity,
Guanaco only requires knowing the fixed value of the submodular width, but does not need to know how this value was computed.}

Recent work introduced a hypergraph measure called $\omega$-submodular width~\cite{omega-submodular} which, by leveraging algorithms for fast matrix multiplication, can for some hypergraphs achieve runtimes strictly below the submodular width.  We believe that Guanaco, with its high simplicity, may be a fruitful footing for further investigating how to combine algebraic methods with the computational primitives used for submodular-width-based algorithms for conjunctive query evaluation.

\section{Preliminaries}
\label{sec:prelims}

When $S$ is a set, we use $\wp(S)$ to denote the power set of $S$.
By a hypergraph, we mean a pair $(V,E)$ consisting
of a set $V$ of vertices and a set $E \subseteq \wp(\rev{V})$ of
edges; when $H$ is a hypergraph, we sometimes denote
$H = (V(H),E(H))$.
Throughout, we use $f \res S$ to denote the restriction of a function $f$ to a set $S$.
We extend this notation to sets of mappings; for a set of mappings $F$,
we use $F \res S$ to denote the set $\{ f \res S ~|~ f \in F \}$.

\subsection{Tree decompositions}
\label{subsec:tree-decompositions}

Let $H$ be a hypergraph.  A \emph{tree decomposition}
of $H$ is a pair $(T,\chi)$ where $T$ is a tree and
$\chi: V(T) \to \wp(V(H))$ is a map satisfying
two properties:
(1) (coverage) for each
$e \in E(H)$, there exists a vertex $t \in V(T)$
with $e \subseteq \chi(t)$, and (2) for each
$v \in V(H)$, the set $\{ t ~|~ v \in \chi(t) \}$
forms a connected subtree of~$T$.
Here, when $t \in V(T)$ is a vertex, $\chi(t)$ is called
the \emph{bag} of $t$.
When $X \subseteq V(H)$, a
\emph{free-connex tree decomposition} of the pair $H, X$
is a tree decomposition of $H$
having the extra property that there is a connected subtree $T'$
of $T$ whose bags contain exactly the vertices in $X$, that is,
where $\bigcup_{t \in V(T')} \chi(t) = X$.
Here, the set $X$ is intended to represent the free variables of a conjunctive query.

\subsection{Submodular width}
\label{subsec:submodular-width}

Relative to a set $V$, a \emph{polymatroid}
is a set function $g: \wp(V) \to \mathbb{R}$
satisfying $g(\emptyset) = 0$;
$X \subseteq Y \subseteq V$ implies $g(X) \leq g(Y)$
(monotonicity);
and
$g(X) + g(Y) \geq g(X \cap Y) + g(X \cup Y)$
for all subsets $X, Y \subseteq V$ (submodularity).
When $H$ is a hypergraph,
let $S(H)$ denote the set containing each polymatroid
$g: \wp(V(H)) \to \mathbb{R}$
satisfying \emph{edge domination}: $g(e) \leq 1$ for each
$e \in E(H)$.
For a hypergraph $H$ and $X \subseteq V(H)$,
the \emph{submodular width} of the pair $H, X$,
denoted by $\subw(H,X)$,
is defined as $\sup_{g \in S(H)} \min_{(T,\chi)} \max_{t \in V(T)} g(\chi(t))$, where here the minimum is over all free-connex tree decompositions $(T,\chi)$ of $H, X$.
In the case that $X = \emptyset$, this becomes a minimum
over all tree decompositions $(T,\chi)$ of $H$, and one obtains Marx's original definition of submodular width~\cite{DBLP:journals/jacm/Marx13}.

\rev{Given a hypergraph $H$ and a set $X \subseteq V(H)$, in order to compute the submodular width, the first step is to distribute the min over the max in the above definition~\cite{DBLP:conf/pods/Khamis0S17,10.1145/3788852.3805595}.
This distribution gives birth to the concept of {\em ``bag selector''}, where a bag selector $B \subseteq \wp(V(H))$ is a set consisting of one bag from each tree decomposition. Let $\mathcal B$ denote the set of all bag selectors.
Then, the submodular width becomes:
\begin{align*}
    \subw(H, X) &= \sup_{g \in S(H)} \quad\min_{(T,\chi)}\quad \max_{t \in V(T)}\quad g(\chi(t))\\
    & = \sup_{g \in S(H)} \quad\max_{B \in \mathcal B}\quad \min_{b \in B}\quad g(b)\\
    & = \max_{B \in \mathcal B} \quad\sup_{g \in S(H)}\quad \min_{b \in B} \quad g(b)\\
    &= \max_{B \in \mathcal B} \quad \underbrace{\max \{z \in \mathbb{R} \mid \exists g \in S(H), z \leq g(b), \forall b \in B\}}_{\text{A linear program}}
\end{align*}
In particular, computing the submodular width reduces to solving a collection of linear programs, one for each bag selector $B$, and taking their maximum.
}

\subsection{Configurations}
\label{subsec:configurations}

To accommodate multiple audiences and multiple backgrounds, we use a formulation of conjunctive query evaluation that, although nonstandard, is intended to be accessible.
Each set of \emph{variables} under discussion will be a finite nonempty set whose subsets are typically denoted with letters such as $S$, $T$, $U$, $\ldots$; each \emph{domain} under discussion will be a finite nonempty set.
A \emph{configuration on a set $V$ of variables and a domain $D$} is a pair $(\R, (F_S)_{S \in \R})$
where
\begin{itemize}
\item
$\R$ is a nonempty subset of $\wp(V)$, that is, a set whose elements are
subsets of $V$,

\item
$(F_S)_{S \in \R}$ is an indexed collection of sets of mappings,
where for each $S \in \R$, it holds that $F_S$ is a set of
mappings from $S$ to $D$.  Such a set $F_S$ is sometimes called a \emph{relation}.

\end{itemize}
We say that a map $f: V \to D$ is an \emph{answer} of a configuration
 $(\R, (F_S)_{S \in \R})$
if for each $S \in \R$,
it holds that $(f \res S) \in F_S$.
We use $\ans(C)$ to denote the set of answers of $C$.
The \emph{hypergraph} of a configuration
$(\R, (F_S)_{S \in \R})$ on a set $V$ of variables and a domain $D$ is the hypergraph $(V,\R)$.

We will discuss two computational problems.  The first one, a formulation of Boolean conjunctive query evaluation, is the problem of deciding, given a configuration, whether or not there exists an answer of the configuration.  The second one, which generalizes the first one and is a formulation of conjunctive query evaluation, is the problem of computing, given a configuration $C$ and a set $X \subseteq V$ of free variables, the set $\ans(C) \upharpoonright X$.
We view the first problem as the case of the second problem where $X = \emptyset$: a configuration $C$ has an answer if and only if $\ans(C) \upharpoonright \rev{\emptyset}$ contains the empty answer, that is, the unique answer whose domain is the empty set $\emptyset$.

%

\begin{example}
    \label{ex:config}
\rev{
Consider, as an example, the conjunctive query
\[Q(U,W,Y,Z) = E(U,W), E(U,Z), G(W,Y,Z).\]
Here, all variables are free, as indicated by the notation
$Q(U,W,Y,Z)$, and so we have $V = X = \{ U, W, Y, Z \}$.
A database $\reld$ for this query would provide a universe $D$ and two relations, namely,~$E^\reld$, which would be a subset of $D^2$, that is, a set of pairs with entries from $D$; and $G^\reld$, which would be a subset of $D^3$, that is, a set of triples with entries from $D$.
In a typical formulation of conjunctive query evaluation on $Q$,
an instance would be a database $\reld$ for this query,
and the goal would be to find the set of query answers, where
$f: V \to D$ is an \emph{query answer} when
\[(f(U),f(W)) \in E^\reld, (f(U),f(Z)) \in E^{\reld}, (f(W),f(Y),f(Z)) \in G^\reld.\]

To formulate such an instance using the notion of configuration,
we would take $\R$ to be $ \{ \{ U,W\}, \{ U, Z \}, \{ W, Y, Z \} \}$,
and define, for example, $F_{ \{U,W \}}$ as the set of maps
$g: \{ U,W \} \to D$ such that $(g(U),g(W)) \in E^{\reld}$; the sets
$F_{ \{ U, Z \}}$ and $F_{\{ W, Y, Z \}}$ would be defined analogously.
Then, an answer of the configuration, as defined above, would be
exactly a query answer as just defined.

As a second example, consider the Boolean conjunctive query
\[Q'() = E(U,W), E(U,Z), G(W,Y,Z).\]
Here, no variable is free, as indicated by the notation $Q'()$,
and so we take  $X = \emptyset$ and $V = \{ U, W, Y, Z \}$.
Given a database $\reld$, we formulate a configuration $C$ from $\reld$
just as above.
Evaluating~$Q'$ on a database $\reld$ now amounts to
deciding if the configuration has any answer at all, that is,
whether or not $\ans(C)$ is nonempty;
we have that $\ans(C)$ is nonempty if and only if $\ans(C) \res \emptyset$ contains the empty answer.
}
\end{example}

We will use the following terminology to discuss configurations.
\rev{A number of the below notions, including the notions of \emph{maximum degree, average degree, and uniformity}, appear in and were used in Marx's article~\cite{DBLP:journals/jacm/Marx13}.}
In what follows, we assume that $M$, $\epsilon$ denote values
where $M \geq 1$ and $\epsilon > 0$.

\begin{itemize}
\item Relative to a configuration $(\R, (F_S)_{S \in \R})$, we refer to the elements of $\R$ as \emph{realized sets}.

\item A configuration $(\R, (F_S)_{S \in \R})$ is \emph{empty} when $\emptyset \in \R$ and $F_{\emptyset} = \emptyset$, and is \emph{non-empty} when $\emptyset \in \R$ and $F_{\emptyset} \neq \emptyset$.  Observe that an empty configuration has no answers.

\item
We say that a configuration $(\R, (F_S)_{S \in \R})$ is \emph{down-closed} if $\R$ is closed
under taking subsets.

\item
A configuration $(\R, (F_S)_{S \in \R})$ is \emph{$\pi$-consistent}
if it is down-closed, non-empty, and
for all $S, T \in \R$ such that $S \subseteq T$, it holds that
$F_T \res S = F_S$.
\item
Let $C = (\R, (F_S)_{S \in \R})$ be a $\pi$-consistent configuration,
and let $\rev{S\subseteq T \in \R}$.
We define the following notions of \emph{maximum} and \emph{average degree}:
$$\maxdeg_C(T|S) := \max_{f \in F_S} |\{ g \in F_T ~|~ \textup{$g$ extends $f$} \} |, \quad
\avgdeg_C(T|S) := |F_T| / |F_S|.$$
Here, we write that $g$ \emph{extends} $f$ when $g \res S = f$.

\item
A down-closed configuration $C = (\R, (F_S)_{S \in \R})$ is \emph{$(M, \epsilon)$-uniform}
if for all $S, T \in \R$ such that $S \subseteq T$, it holds that
$\maxdeg_C(T|S) \leq M^{\epsilon} \avgdeg_C(T|S)$.

\item
A down-closed configuration is
\emph{$(M, \epsilon)$-clean} if it is $\pi$-consistent
and $(M,\epsilon)$-uniform.

\end{itemize}

Note that in general, we will only discuss uniformity of configurations
that are $\pi$-consistent.

%

%

We say that a set $\C'$ of configurations \emph{represents} a configuration
$C$ if $\ans(C)$ is equal to $\bigcup_{C' \in \C'} \ans(C')$.
Say that a configuration
$C' = (\R', (F'_S)_{S \in \R'})$
\emph{tightens} or
\emph{is a tightening of} a configuration
$C = (\R, (F_S)_{S \in \R})$
if for each $S \in \R$, it holds that $S \in \R'$ and
$F'_S \subseteq F_S$.
That is, if~$C$ holds a relation on $S$, then $C'$ also holds a relation on $S$ which must be \rev{contained in that of $C$}.
Observe that when $C'$ tightens $C$, it holds that
$\ans(C') \subseteq \ans(C)$.

We will use the known fact concerning $\pi$-consistency.
\rev{Here and throughout, we assume a RAM model, and we study a data complexity, so (for example) the variable set under discussion is viewed as having constant size.}

\begin{prop}
\label{prop:establish-pi-consistency}
Fix a variable set $V$.
There exists a linear-time algorithm that, given a configuration $C$ on variable set $V$, \emph{establishes $\pi$-consistency} in that it removes functions from $C$'s relations so that the resulting configuration $C'$ is $\pi$-consistent or empty, is a tightening of $C$, and also
has $\ans(C') = \ans(C)$.
\end{prop}

\begin{proof}
In brief, the algorithm acts as follows.
First, it down-closes the input configuration, if needed, by adding relations that are projections of existing ones; this can be done in linear time and without changing the answer set.  Then, it continually removes functions from relations until either $\pi$-consistency holds, or all relations become empty.
This removal process can be implemented in linear time,
by reduction to the problem of finding a minimum weight satisfying assignment to a set of Horn clauses, as explained in the proof of Lemma 4.2~\cite{BerkholzSchweikardt2020-arxiv}. 
\end{proof}

\begin{note}
\label{note:yannakakis}
Fix a hypergraph $H$ and a subset $X \subseteq V(H)$.
We will use the known fact that there is a version of the \rev{Yannakakis algorithm~\cite{DBLP:conf/vldb/Yannakakis81,AbiteboulHullVianu95-foundationsdatabases}}, denoted $\mathsf{Yannakakis}(C', X, \mathcal{T})$ that, given a $\pi$-consistent configuration $C'$ with hypergraph $H$ and a free-connex tree decomposition $\mathcal{T} = (T,\chi)$ of $H,X$, where all bags are realized in $C'$, acts as follows:
it restricts the tree decomposition to a subtree~$T'$ where
$\bigcup_{t \in V(T')} \chi(t) = X$, and then enumerates
the
set $\ans(C^-) \res X$ where $C^-$ is the configuration
obtained from $C'$ where
the realized sets are restricted to the sets $\chi(t)$.
It does this within time $O(N + \out')$
where $N$ is the size of $C'$, and $\out'$ is the size of the output set.
Clearly, we have
$(\ans(C') \res X) \subseteq (\ans(C^-) \res X)$.
\end{note}

\rev{
\section{Intuition}
\label{sec:intuition}
Before explaining the $\guanaco$ algorithm in detail, we present in this section the high-level intuition behind the algorithm and give an example execution.
To that end, consider the following 4-cycle query:
\begin{equation}
\label{eq:4cycle}
Q_\square() = R(X,Y), S(Y,Z), T(Z,W), U(W,X).
\end{equation}
It is well-known that $Q_\square$ has a {\em fractional hypertree width}~\cite{DBLP:conf/soda/GroheM06} of 2 and a submodular
width of 3/2.
This query has two free-connex tree decompositions, apart from the trivial one that consists of a single bag containing all four variables.
These two non-trivial tree decompositions are shown in Figure~\ref{fig:4cycle}, and they are as follows:
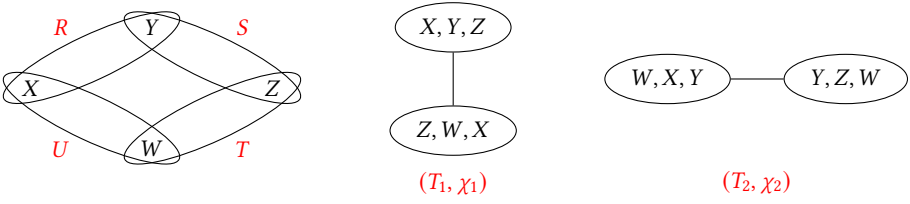
\begin{figure}[t]
    \centering
    \begin{tikzpicture}[scale = .4, every node/.style={scale=0.9}]
        \node[] at (0,0) (X) {$X$};
        \node[] at (8,0) (Z) {$Z$};
        \node[] at (4,2) (Y) {$Y$};
        \node[] at (4,-2) (W) {$W$};
        \draw[rotate around={+25:(2, +1)}, ] (2, +1) ellipse (3.2 and .75)
        node[shift={(-.45cm,+.45cm)},red]{$R$};
        \draw[rotate around={-25:(6, +1)}, ] (6, +1) ellipse (3.2 and .75)
        node[shift={(+.45cm,+.45cm)},red]{$S$};
        \draw[rotate around={+25:(6, -1)}, ] (6, -1) ellipse (3.2 and .75)
        node[shift={(+.45cm,-.45cm)},red]{$T$};
        \draw[rotate around={-25:(2, -1)}, ] (2, -1) ellipse (3.2 and .75)
        node[shift={(-.45cm,-.45cm)},red]{$U$};

        \node[draw,ellipse, right = 3 of Y] (xyz) {$X, Y, Z$};
        \node[draw,ellipse, below = 0.7 of xyz] (zwx) {$Z, W, X$};
        \draw (xyz) -- (zwx);
        \node[below = .1 of zwx] {\color{red}$(T_1,\chi_1)$};

        \coordinate (TD1) at ($(xyz)!0.5!(zwx)$);
        \node[draw,ellipse, right = 2 of TD1] (wxy) {$W, X, Y$};
        \node[draw,ellipse, right = 0.7 of wxy] (yzw) {$Y, Z, W$};
        \draw (wxy) -- (yzw);
        \coordinate (TD2) at ($(wxy)!0.5!(yzw)$);
        \node[below = 1.1 of TD2] {\color{red}$(T_2,\chi_2)$};
    \end{tikzpicture}
\caption{\rev{The 4-cycle query $Q_\square$ (Eq.~\eqref{eq:4cycle}) and the two non-trivial tree decompositions.}}
\label{fig:4cycle}
\end{figure}
\begin{itemize}
    \item A tree decomposition $(T_1,\chi_1)$ with two nodes $t_{11}$ and $t_{12}$ corresponding to two bags $\chi_1(t_{11}) = \{X, Y, Z\}$ and $\chi_1(t_{12}) = \{Z, W, X\}$.
    \item A tree decomposition $(T_2,\chi_2)$ with two nodes $t_{21}$ and $t_{22}$ corresponding to two bags $\chi_2(t_{21}) = \{W, X, Y\}$ and $\chi_2(t_{22}) = \{Y, Z, W\}$.
\end{itemize}
The fact that the fractional hypertree width of $Q_\square$ is 2 indicates that,
no matter which tree decomposition we choose,
there is a worst-case database instance where we cannot beat the $O(N^2)$ runtime, where $N$ is the input size.
Indeed, here is such an instance. Assume $N$ is an even number, and let $[N]$ denote the set $\{1, \ldots, N\}$. Let the relations $R, S, T, U$ all be identical
and contain the following tuples:
\begin{align*}
    R = S = T = U = ([N/2]\times \{1\}) \cup (\{1\}\times [N/2])
\end{align*}
In the configuration notation, the input is a configuration $C_0 = (\R_0, (F^0_S)_{S \in \R_0})$ where the domain $D=[N/2]$, $\R_0 = \{ \{X,Y\}, \{Y,Z\}, \{Z,W\}, \{W,X\} \}$,
$F^0_{\{X, Y\}}$ is defined as follows:
\begin{align}
    F^0_{\{X, Y\}} = \{f : \{X, Y\} \to D \mid (f(X), f(Y)) \in ([N/2]\times \{1\}) \cup (\{1\}\times [N/2])\},
    \label{eq:4cycle-input}
\end{align}
and $F^0_{\{Y, Z\}}$, $F^0_{\{Z, W\}}$, and $F^0_{\{W, X\}}$ are defined similarly.

It is not hard to see that no matter which tree decomposition we choose,
the intermediate relations corresponding to the bags will have size $\Theta(N^2)$.
However, it is well-known that this query admits an algorithm with runtime $O(N^{3/2})$~\cite{DBLP:journals/algorithmica/AlonYZ97} on any database instance
(where the exponent $3/2$ happens to be precisely the submodular width of $Q_\square$).
In contrast, our $\guanaco$ algorithm achieves a runtime of $O(N^{3/2+\epsilon})$ for any $\epsilon > 0$ on any database instance.
On the particular database instance specified above,
$\guanaco$ achieves a better runtime of $O(N)$.
We explain below how this is possible.

$\guanaco$ uses a subroutine, called $\cleanup$, that takes any configuration
$C = (\R, (F_S)_{S \in \R})$ and
partitions it into a constant number (in data complexity) of configurations that are $(M, \epsilon)$-clean, for some given parameters $M$ and $\epsilon$.
Recall that a configuration is $(M, \epsilon)$-clean if it is $\pi$-consistent and $(M, \epsilon)$-uniform.
Ensuring $\pi$-consistency is straightforward:
In the classic language of databases, this is equivalent to performing {\em semijoin-reductions} on the input relations until no further reductions are possible,
and this can be done in linear time (by Proposition~\ref{prop:establish-pi-consistency}).
On the other hand, ensuring $(M, \epsilon)$-uniformity is more involved.
Recall that uniformity means that for any two realized sets $S \subseteq T$ in the configuration, the maximum degree $\maxdeg_C(T|S)$ is not significantly larger than the average degree $\avgdeg_C(T|S)$.
Intuitively, this means that the degree $T|S$ should not have a very high variance, hence the name ``uniformity''.
Whenever we detect a violation of uniformity for some $S \subseteq T$,
the $\cleanup$ subroutine partitions the configuration $C$ into two configurations
$C^\ell$ and $C^h$ by partitioning the relation $F_S$ into two parts:
A ``light part'' $F_S^\ell$ where the degree of each tuple is ``small'',
and a ``heavy part'' $F_S^h$ where the degree of each tuple is ``large''.
Fixing the uniformity violation for $S \subseteq T$ might result in the new configurations $C^\ell$ and $C^h$ violating $\pi$-consistency, and similarly, fixing $\pi$-consistency might result in re-violating previously fixed uniformity violations.
Therefore, the $\cleanup$ subroutine iteratively alternates between both until both properties are satisfied, and returns the resulting configurations.
It is not obvious that this process terminates.
We will see in Section~\ref{sec:cleanup} that the $\cleanup$ subroutine does indeed terminate and return a constant number of $(M, \epsilon)$-clean configurations.
In our example, the $\cleanup$ subroutine will partition the input configuration $C_0$
by partitioning $F^0_{\{X, Y\}}$ into $F_{\{X, Y\}}^\ell$ and $F_{\{X, Y\}}^h$ that are defined as follows:
\begin{align}
    F_{\{X, Y\}}^\ell &= \{f : \{X, Y\} \to D \mid (f(X), f(Y)) \in [N/2]\times \{1\} \},\\
    F_{\{X, Y\}}^h &= \{f : \{X, Y\} \to D \mid (f(X), f(Y)) \in \{1\}\times [N/2]\}.
\end{align}
This is because $X$-values in $F^0_{\{X, Y\}}$ have a maximum degree of $N/2$ and an average degree of $O(1)$.
Hence, the $X$-value of $1$ is ``heavy'' and the other $X$-values are ``light''.
The other relations $F^0_{\{Y, Z\}}$, $F^0_{\{Z, W\}}$, and $F^0_{\{W, X\}}$ are partitioned similarly.
At the end, the $\cleanup$ subroutine will return at most 16 configurations, capturing the combinations where each of the four variables $X, Y, Z, W$ is either heavy or light.
(The actual number of configurations might vary depending on some implementation details of $\cleanup$, but it will always be a constant number in data complexity.)

The outer $\guanaco$ algorithm is based on the following idea.
Consider an instance, represented as a configuration $C_0 = (\R_0, (F^0_S)_{S \in \R_0})$ of size $N$, for example, the configuration described in Eq.~\eqref{eq:4cycle-input} for the 4-cycle query.
Consider any configuration $C = (\R, (F_S)_{S \in \R})$ that tightens $C_0$ and is $(N, \epsilon)$-clean.
For simplicity, suppose that $\epsilon=0$, hence $C$ is ``perfectly uniform'',
meaning that, for any two realized sets $S \subseteq T$, the average degree and maximum degree are the same.
In our example, $C$ could be one of the 16 configurations returned by the $\cleanup$ subroutine.
Let $w$ denote the submodular width of the query $Q_\square$, which is 3/2 in this example.
Suppose that all realized relations $F_S$ in $C$ satisfy $|F_S| \leq N^w$.
Consider the following set function
$g: \wp(\{X, Y, Z, W\}) \to \mathbb{R}$:
\begin{align}
g(U) =
\begin{cases}
\log_N |F_U| & \text{if $U \in \R$, that is, $U$ is realized (which implies $\log_N|F_U| \leq w$)} \\
w & \text{otherwise } \\
\end{cases}
\end{align}
Then, $g$ is a polymatroid.
Indeed, the monotonicity of $g$ follows from the fact that $C$ is $\pi$-consistent, and the submodularity of $g$ follows from the fact that $C$ is $(N, 0)$-uniform.
To see why submodularity holds, consider any two realized sets $S, T \in \R$ where $S\cup T$
is realized (hence down-closedness implies that $S \cap T$ is realized as well).
Then, we have:
\begin{align*}
    |F_{S\cup T}| \leq |F_S| \cdot \maxdeg_C(T|S \cap T) = |F_S| \cdot \avgdeg_C(T|S \cap T) = |F_S| \cdot \frac{|F_T|}{|F_{S\cap T}|},
\end{align*}
By taking logarithms, we obtain the submodularity $g(S\cup T) + g(S\cap T) \leq g(S) + g(T)$.
Other submodularity cases can be verified similarly (under mild conditions).
Moreover, since $C$ tightens $C_0$, the polymatroid $g$ is edge-dominated.
But then, recall from Section~\ref{subsec:submodular-width}
that the submodular width is a maximum over all edge-dominated polymatroids, including the above $g$.
We will see later in Section~\ref{sec:guanaco} that this implies either one of two things:
\begin{itemize}
    \item Either there exists a free-connex tree decomposition of the query
    whose bags are already realized in $C$, in which case we can answer the query by running the Yannakakis algorithm on $C$.
    \item Or every realized relation $F_S$ satisfies $\log_N |F_S| \leq w$ {\em and} there exists a ``realizable pair'' $(S, T)$, i.e.,
    a pair of realized sets $(S, T)$ where $S \cup T$ is not realized in $C$ and $g(S) + g(T) - g(S \cap T) \leq w$.
    If this is the case, then we can join the two relations $F_S$ and $F_T$
    to create a new relation $F_{S\cup T}$, and add it to the configuration.
    The size of the new relation $F_{S\cup T}$ is upper bounded by $|F_S|\cdot|F_T|/|F_{S\cap T}|$. On $\log_N$-scale, this is equivalent to $g(S) + g(T) - g(S \cap T)\leq w$, hence $|F_{S\cup T}|\leq N^w$.
    Adding this extra relation $F_{S \cup T}$ to the configuration might violate $\pi$-consistency and $(N, \epsilon)$-uniformity, which is why we need to run $\cleanup$ again.
\end{itemize}
In our example, the $\guanaco$ algorithm initially calls the $\cleanup$
subroutine on the the input configuration $C_0 = (\R_0, (F^0_S)_{S \in \R_0})$ described by Eq.~\eqref{eq:4cycle-input} and forth, and obtains a constant number of $(N, \epsilon)$-clean configurations, as explained before.
WLOG consider the configuration $C'=(\R_0, (F'_S)_{S \in \R_0})$ where $Y$ and $W$ are ``light'' and $X$ and $Z$ are ``heavy''.
In this configuration $C'$, the realized relations are as follows:
\begin{align}
    F'_{\{X, Y\}} &= \{f : \{X, Y\} \to D \mid (f(X), f(Y)) \in \{1\} \times [N/2] \},\\
    F'_{\{Y, Z\}} &= \{f : \{Y, Z\} \to D \mid (f(Y), f(Z)) \in [N/2] \times \{1\}\},\\
    F'_{\{Z, W\}} &= \{f : \{Z, W\} \to D \mid (f(Z), f(W)) \in \{1\} \times [N/2]\},\\
    F'_{\{W, X\}} &= \{f : \{W, X\} \to D \mid (f(W), f(X)) \in [N/2] \times \{1\} \}.
\end{align}
The corresponding polymatroid $g$ is as follows:
\begin{align}
    g(U) &= \min(3/2, |U \cap \{Y, W\}|), \quad\quad \text{for all } U \subseteq \{X, Y, Z, W\}.
\end{align}
This is indeed a polymatroid that is edge-dominated.
Since no tree decomposition from Figure~\ref{fig:4cycle} already has both bags
realized in $C'$,
it must be the case that $C'$ contains at least one realizable pair.
In fact, $C'$ contains two realizable pairs in this case:
\begin{itemize}
    \item The pair $(\{X, Y\}, \{Y, Z\})$ where $g(\{X, Y\}) + g(\{Y, Z\}) - g(\{Y\}) = 1 + 1 - 1 = 1\leq 3/2$. We can join the two relations $F'_{\{X, Y\}}$ and $F'_{\{Y, Z\}}$ to create a new relation $F'_{\{X, Y, Z\}}$ of size $O(N)$.
    \item The pair $(\{Z, W\}, \{W, X\})$ where $g(\{Z, W\}) + g(\{W, X\}) - g(\{W\}) = 1 + 1 - 1 = 1\leq 3/2$. We can join the two relations $F'_{\{Z, W\}}$ and $F'_{\{W, X\}}$ to create a new relation $F'_{\{W, X, Z\}}$ of size $O(N)$.
\end{itemize}
After creating the above two relations, we can run the Yannakakis algorithm using the tree decomposition $(T_1, \chi_1)$.
Hence, the configuration $C'$ takes time $O(N)$ overall.
The other 15 configurations take similar time and use either one of the two tree decompositions.
}

\section{Cleanup subroutine}
\label{sec:cleanup}

In this section, we introduce a basic subroutine called \emph{Cleanup} that interleaves establishing $\pi$-consistency with establishing global uniformity until clean configurations are obtained.  This subroutine and its analysis can be viewed as formulations of similar ones due to Marx~\cite{DBLP:journals/jacm/Marx13}.\footnote{\rev{Indeed, Lemma~\ref{lemma:cleanup} below uses the same core argumentation as Lemma 4.11 in Marx~\cite{DBLP:journals/jacm/Marx13}.  The notion of uniformity employed by Marx is different from that employed here.}}
More precisely, relative to a value $\epsilon > 0$ and a value $M \geq 1$, the \emph{Cleanup} algorithm, when
given a configuration
$C$ over variables $V$, outputs a finite set of down-closed configurations
that represent~$C$, have the same realized sets as $C$ \rev{(up to down-closing)}, and are each $(M,\epsilon)$-clean.
Moreover, there exists a computable function $g$ such that
the algorithm runs within time $O(g(|V|, \epsilon) \cdot N)$,
where~$N$ denotes the input database size.
The algorithm description uses the following notions,
relative to a configuration
$C~=~(\R, (F_S)_{S \in \R})$.
Say that a pair $S,T$ \emph{violates}
$(M,\epsilon)$-uniformity in $C$
if $S \subseteq T \in \R$ and
$\maxdeg_C(T|S) > M^{\epsilon} \avgdeg_C(T|S)$.
When $S \subseteq T \in \R$, and $f \in F_S$,
define  $\deg_C(T|f)$ as $|\{ g \in F_T ~|~ \textup{$g$ extends $f$} \}|$,
that is, $\deg_C(T|f)$ is the number of extensions of $f$ in $F_T$.


\begin{algorithm}
    \caption{
$\cleanup_{\epsilon}$(value $M$, configuration $C$); assumes $\epsilon > 0, M \geq 1$
}
    \label{algo:cleanup}
    \begin{algorithmic}[1]
    	\State Establish $\pi$-consistency on $C$ (using the algorithm of Proposition~\ref{prop:establish-pi-consistency})
    	\If{$C$ is empty}
    		\State Return the empty set $\emptyset$
    	\EndIf
        \If{there exists a pair violating $(M,\epsilon)$-uniformity in $C$}
            \State Let $S$, $T$ be such a pair                        
            \State
Assign $F_S^\ell :=
\{ f \in F_S ~|~ \deg_C(T|f) \leq \sqrt{M^{\epsilon}} \avgdeg_C(T|S) \}$

            \State
Assign $F_S^h :=
\{ f \in F_S ~|~ \deg_C(T|f) > \sqrt{M^{\epsilon}} \avgdeg_C(T|S) \}$
			\State
Assign $C^\ell$ and $C^h$ to be equal to $C$ but where
$F_S$ is replaced by $F_S^\ell$ and $F_S^h$, respectively
			\State Return $\cleanup(M, C^\ell) \cup \cleanup(M, C^h)$
		\Else
			\State Return $\{ C \}$ \label{alg:cleanup-return}
        \EndIf
   \end{algorithmic}
\end{algorithm}

\begin{lemma}
\label{lemma:cleanup}
Let $\epsilon > 0, w^+ \geq 1$ be constants; there is a function $d(\cdot) \in O(3^n)$ such that
the following holds.
Assume that
$\cleanup_{\epsilon}$ is run on $M \geq 1$ and
a configuration $C$ where
every relation has size at most~$M^{w^+}$;
then,
the length of any root-leaf path in the recursion tree (that is, the depth of the recursion tree) is bounded above by $d(|V|)$, where $V$ is the variable set of $C$.
In particular, the algorithm terminates.
\end{lemma}

\begin{proof}
In the case that $M=1$, it is evident that the algorithm does not recurse and that the lemma holds, so for the remainder of the proof, we assume that $M > 1$.

Let $C^0, \ldots, C^p$ be the configurations of any root-leaf path in the recursion tree; precisely, we consider the configurations right after $\pi$-consistency is established.
Observe that all configurations ever processed by the algorithm share the same set $\R$ of realized sets.
Let us denote $C^j = (\R, (F^j_S)_{S \in \R})$.
We define the potential function
\[W^j = \sum_{S, T: S \subseteq T \in \R} \log_M \maxdeg_{C^j}(T|S). \]
By assumption, each relation has size at most $M^{w^+}$, so each term of the sum is at most $w^+$.  The number of terms is at most $3^{|V|}$, so we have $W^0 \leq 3^{|V|} w^+$.
We show that in each recursive call, this function decreases by a sufficient amount so as to allow us to bound the path length as desired.

Suppose $j$ is an index with $0 \leq j < p$.
We claim that $W^{j+1} \leq W^j - \epsilon/2$.
We consider two cases depending on how $C^{j+1}$
is derived from $C^j$.
\begin{itemize}

\item If $C^{j+1}$ is derived from $(C^j)^\ell$,
we have
\begin{align*}
\log_M \maxdeg_{C^{j+1}} (T|S)
& \leq \log_M \maxdeg_{(C^j)^\ell} (T|S) \\
& \leq \log_M (M^{\epsilon/2} \avgdeg_{C^j}(T|S)) \\
& <    \log_M (M^{\epsilon/2} \frac{\maxdeg_{C^j}(T|S)}{M^{\epsilon}}) \\
& = (\log_M \maxdeg_{C^j}\rev{(T|S)}) - (\epsilon/2).
\end{align*}
The first inequality follows from the fact that establishing $\pi$-consistency only removes tuples from relations and cannot increase maximum degree; the second inequality comes from the definition of $(C^j)^\ell$ from $C^j$, and the third inequality comes from $S,T$ being a violation of uniformity.

\item
If $C^{j+1}$ is derived from $(C^j)^h$,
we first observe that each $f \in (F_S^j)^h$
has at least \\
$\sqrt{M^{\epsilon}} \avgdeg_{C^j}(T|S) =
 \sqrt{M^{\epsilon}} \frac{|F^j_T|}{|F^j_S|}$
extensions in $F_T^j$.  We then have
$$|(F_S^j)^h| \leq |F_T^j| / (\sqrt{M^{\epsilon}} \frac{|F^j_T|}{|F^j_S|})  = \frac{|F^j_S|}{\sqrt{M^\epsilon}}.$$

It follows that
\begin{align*}
\log_M \maxdeg_{C^{j+1}} (S|\emptyset)
& \leq \log_M \maxdeg_{(C^j)^h} (S|\emptyset) \\
& = \log_M |(F_S^j)^h| \\
& \leq (\log_M |F^j_S|) - (\epsilon/2) \\
& = (\log_M \maxdeg_{C^j} (S|\emptyset)) - (\epsilon/2).
\end{align*}
Here, the first inequality is justified in the same way as the first inequality in the previous case, and the second inequality follows from the observation.
\end{itemize}

It follows from the claim that the potential function decreases by at least $\epsilon/2$ at each stage, and hence we have
$p \leq \frac{3^{|V|} w^+}{\epsilon/2}$;
for fixed $\epsilon$ and $w^+$, this bound
is $O(3^{|V|})$.
\end{proof}

The following theorem summarizes the main properties of the $\cleanup$ algorithm to be used.

\begin{theorem}
\label{thm:cleanup}
Fix constants $\epsilon > 0, w^+ \geq 1$, and a variable set $V$.
Assume that the $\cleanup_{\epsilon}$ algorithm is invoked on $M \geq 1$ and a configuration $C$ on variable set $V$
where each relation has size at most $M^{w^+}$;
the algorithm returns a set
$\C$ of configurations representing $C$ and having size
$|\C| \leq 2^{O(3^{|V|})}$, where
each configuration in $\C$ is $(M,\epsilon)$-clean, tightens $C$, and has the same realized sets as $C$ does.
Moreover, the algorithm admits an implementation running within time $O(N)$.
\end{theorem}

\begin{proof}
Since, by Lemma~\ref{lemma:cleanup}, the depth of the recursion tree is bounded by $O(3^{|V|})$, the size of the tree is bounded by $2^{O(3^{|V|})}$; this latter value in turn bounds the size of the output set $\C$, which contains one configuration per leaf of the recursion tree.
The set $\C$ represents the configuration~$C$ due to the fact that when $C$ is partitioned into $C^\ell$ and $C^h$, we evidently have
$\ans(C) = \ans(C^\ell) \cup \ans(C^h)$.
Each configuration in $\C$ is $(M,\epsilon)$-clean, since whenever line~\ref{alg:cleanup-return} is reached, the configuration returned is both $\pi$-consistent and has no pair violating $(M,\epsilon)$-uniformity; also, each such configuration evidently tightens $C$ and has the same realized sets as $C$ does.
The main computations performed at each node of the recursion tree are the step that establishes $\pi$-consistency, which runs in linear time due to Proposition~\ref{prop:establish-pi-consistency}, and the partitioning of a relation, which is evidently in linear time.
\end{proof}

\section{Guanaco algorithm}
\label{sec:guanaco}

We present our main algorithm and its analysis.
Let us first give a brief intuitive description of the
algorithm.  The main algorithm $\guanaco$ receives
a configuration $C_0$, and computes some values.
It then calls the main recursive subroutine,
called $\guanacorec$, on $C_0$.
This subroutine calls $\cleanup$ once,
and then for each resulting configuration $C'$,
it tries to \emph{realize} a new relation if possible;
if this is possible, it then calls itself recursively
on the consequent configuration, and
in the case that this is not possible,
it finds a tree decomposition and runs the Yannakakis
algorithm to enumerate the needed answers.
(Note that in this case, although we present the algorithm $\guanacorec$ as attempting to find a tree decomposition of a certain type, it will turn out that such a tree decomposition will always be found.)
Thus, on a high level, this subroutine interleaves $\cleanup$ with the creation of new relations.

How are new relations created?  The following
definition will govern when this is done.
Say that a pair $(S,T)$ is \emph{$(M,w^+)$-realizable},
with respect to a configuration
$C' = (\R', (F'_S)_{S \in \R'})$,
if the following hold: $S, T \in \R'$, $S \cup T \notin \R'$, and
\rev{$\log_M(|F'_S| \cdot \maxdeg_{C'}(T | S \cap T)) \leq w^+$.}

\begin{note}
\label{note:realization}
Let $(S,T)$ be an $(M,w^+)$-realizable pair
with respect to a configuration $C'$.
We speak of \emph{realizing} such a pair $(S,T)$
on configuration $C'$, which means that $S \cup T$
is added to $\R'$ and a relation $F'_{S \cup T}$ is created,
by considering each pair of assignments
$f \in F'_S$ and $g \in F'_T$ where $g$ extends $f \res S \cap T$,
and including (for each such pair of assignments) the assignment $f \cup g$ in
$F'_{S \cup T}$.
(Here, by $f \cup g$, we are viewing assignments as sets of pairs.)
The number of pairs $f,g$ that have to be considered, by definition of $(M,w^+)$-realizable, is less than or equal to $M^{w^+}$, as is the size of the resulting relation $F'_{S \cup T}$.  Thus in effect, this process of \emph{realizing} a pair amounts to taking a binary join whose output size is controlled.
\end{note}

\begin{algorithm}
    \caption{
$\guanaco_{H,X,\delta}$(configuration $C_0$); assumes $\delta > 0$,
$C_0$ has hypergraph $H$, and $X \subseteq V(H)$}
    \label{algo:guanaco}
    \begin{algorithmic}[1]
    	\State
Assign $w := \subw(H, X)$ \algorithmiccomment{\rev{See Section~\ref{subsec:submodular-width} for how to compute this}}
		\State Assign $w^+ := w(1+\delta)$
		\State Assign $\epsilon := \min(( 1 - \frac{1}{1+\delta})^{4}, \frac{1}{|V|^{4}} )$, where $V$ is the variable set of $C_0$
    	\State Assign $M$ to be the maximum relation size of $C_0$
    	\If{$M$ = 0}
    		\State Return the empty set $\emptyset$
    		\label{alg:guanaco-emptyset}
    	\EndIf
    	\State Return $\guanacorec_{H, X, M,\epsilon,w^+}(C_0)$
    \end{algorithmic}
\end{algorithm}

\begin{algorithm}
    \caption{
$\guanacorec_{H, X, M,\epsilon,w^+}$(configuration $C$); assumes $M \geq 1, w^+ \geq 1$; $\epsilon > 0$; and that $H, X$ are as in $\guanaco_{H,X,\delta}$}
    \label{algo:guanacorec}
    \begin{algorithmic}[1]
    	\State Assign $\C' := \cleanup_{\epsilon}(M, C)$ \label{alg:gr-cleanup}
    	\State Initialize $A := \emptyset$
    	\For{each configuration $C' \in \C'$}
    		\If{$C'$ has an $(M,w^+)$-realizable pair}
    		\label{alg:gr-check-realizable}
    			\State Let $C''$ be the configuration resulting from realizing such a pair on $C'$
    			\label{alg:gr-realize}
    			\State Assign $A := A \cup \guanacorec_{H, X, M,\epsilon,w^+}(C'')$
    			\label{alg:guanacorec-recurse}
    		\Else
    			\State Search for a free-connex tree decomp.  $\mathcal{T}$ of $H,X$ where all bags are realized, wrt $C'$
    			\label{alg:gr-find-td}
    			\State Assign $A := A \cup \mathsf{Yannakakis}(C', X, \mathcal{T})$ \quad (assuming such a TD $\mathcal{T}$ was found)    			\label{alg:gr-yannakakis}
    		\EndIf
    	\EndFor
    	\State Return $A$
    	\label{alg:gr-return-at-end}
    \end{algorithmic}
\end{algorithm}

We first show termination of the main algorithm, along with a time bound.
Here and throughout this section, we use $\out$ to denote the output size.

\begin{lemma}
\label{lemma:guanaco-termination}
(termination/time bound)
Fix a hypergraph $H$ with vertex set $V$;
fix $X \subseteq V$ and $\delta > 0$.
On any input configuration $C_0$ whose hypergraph is $H$,
the call $\guanaco_{H,X,\delta}(C_0)$ terminates
and runs within time
$O(N^{\subw(H,X) + \delta} + \out)$.
\end{lemma}

\begin{proof}
First, let us consider the recursion tree that $\guanacorec$ gives rise to, and bound its depth.  Each time a recursive call is made, a new relation has been realized, by line~\ref{alg:gr-realize}, and the call to $\cleanup_{\epsilon}(M, C)$ results in a set of configurations with the same realized sets as $C$ (by Theorem~\ref{thm:cleanup}); thus,
the maximum depth is $2^{|V|}$.
The branching factor of the recursion tree is determined by the size of the output given by
the call to $\cleanup_{\epsilon}(M, C)$ in line~\ref{alg:gr-cleanup}; but this is a function of $|V|$, by Theorem~\ref{thm:cleanup}.
The size of the entire recursion tree is thus bounded above by a function of $|V|$.
The algorithm thus terminates.

To argue the claimed time bound, it suffices to argue that each individual execution of the routine $\guanacorec$ (that is, excluding the time incurred by recursive calls) runs within time
$O(N^{\subw(H) + \delta} + \out)$.
By Theorem~\ref{thm:cleanup},
line~\ref{alg:gr-cleanup} runs within linear time
$O(N)$, as do
line~\ref{alg:gr-check-realizable},
the check for a realizable pair, and
line~\ref{alg:gr-find-td},
the search for a tree decomposition.
Line~\ref{alg:gr-yannakakis} executes
a call to the Yannakakis algorithm;
as remarked in Note~\ref{note:yannakakis}
it is known that
this algorithm runs within time $O(N + \out')$,
where $\out'$ here denotes the output size of
$\ans(C') \res X$; but $\ans(C') \res X$ forms part of
the overall output, so we have $\out' \leq \out$.
The only remaining line to analyze is
line~\ref{alg:gr-realize};
but as explained in Note~\ref{note:realization},
realizing an $(M, w^+)$-realizable pair
can be performed within time $M^{w^+}$.
\end{proof}

Towards showing correctness of the algorithm, the following lemma is key; somewhat similar lemmas were given by Marx~\cite{DBLP:journals/jacm/Marx13} and by Berkholz and Schweikardt~\cite{berkholz_et_al:LIPIcs.MFCS.2019.58}.  Our lemma will allow us to derive that, whenever $\guanacorec$ cannot realize an additional relation on a configuration $C' \in \C'$, there is a submodular function $g$  where every set $S$ with $g(S) \leq w$ is realized in the configuration $C'$, and moreover has the properties needed to guarantee in line~\ref{alg:gr-find-td} that there is a
tree decomposition of the desired form.
Speaking \emph{approximately}, the function is defined by setting~$g(S)$ to a value close to $\log_M |F_S|$ for each realized set $S$, and to a value close to $w^+$ for each non-realized set $S$.

\begin{lemma}
\label{lemma:submodular}
Suppose that $\delta> 0$, $w\geq 1$, and $M \geq 1$ are values, $V$ is a set of variables, and
$\epsilon$, $w^+$ are defined as in the $\guanaco$ algorithm.
Suppose that a non-empty configuration $C = (\R, (F_S)_{S \in \R})$ with variable set $V$
is such that:
\begin{itemize}

\item it is $(M,\epsilon)$-clean,
\item it has no $(M,w^+)$-realizable pair,
\item every realized set $S$ has $|F_S| \leq M^{w^+}$.

\end{itemize}
Then there exists a monotone submodular function $g: \wp(V) \to \Reals$ such that:
\begin{enumerate}
\item $g(\emptyset) = 0$.
\item Each set $S$ of variables satisfying
$g(S) \leq w$ is realized, that is, is in $\R$.
\item Each set $S \in \R$ where $|F_S| \leq M$
satisfies $g(S) \leq 1$.
\end{enumerate}
\end{lemma}

\begin{proof}
\rev{We show that the following function $g$ has the described properties.
 Set $\alpha = 1 - \epsilon^{1/3}$, and define the function $h: \wp(V) \to \Reals$ by
$h(U) = 2\epsilon^{2/3}|U| - \epsilon|U|^2$.
Define
\[
g^-(U) =
\begin{cases}
\alpha \log_M |F_U| & \text{if $U \in \R$, that is, $U$ is realized} \\
\alpha w^+ & \text{otherwise } \\
\end{cases}
\]
Define $g(U)$ as $g^-(U) + h(U)$.
In the particular case that $M=1$, each instance of
a $\log_M$ operation will only be applied to the value $1$, in this proof; we here use the convention that $\log_1 1 = 1$.
}

We first observe that $\epsilon^{1/3} |V| < 1$,
so
any set $U \subseteq V$ of variables
has $0 \leq \epsilon^{1/3} |U| < 1$;
we have
$h(U) = \epsilon^{1/3} [\epsilon^{1/3} |U| (2 - \epsilon^{1/3} |U|)]$, and can conclude
$0 \leq h(U) \leq \epsilon^{1/3}$.

We argue monotonicity of $g$ as follows.
The function $h$ is readily verified to be monotone.
It thus suffices to verify that $g^-$ is monotone.  But this follows
from the following: first, for any realized sets $S$, $T$ with $S \subseteq T$,
by $\pi$-consistency we have $\log_M |F_S| \leq \log_M |F_T|$ and thus $g^-(S) \leq g^-(T)$;
second, by assumption, for any realized set $T$, we have
$\log_M |F_T| \leq w^+$, so if a set $U$ is not realized, we have $g^-(T) \leq g^-(U)$.

We argue submodularity of $g$ as follows.
First, we note that it is elementarily verified, \rev{using the identity
$|S \cap T|^2 + |S \cup T|^2 - |S|^2 - |T|^2 = 2 \cdot  |S \setminus T| \cdot |T \setminus S|$,} that
\begin{equation}
h(S) + h(T) - h(S \cap T) - h(S \cup T)
= 2\epsilon \cdot |S \setminus T| \cdot |T \setminus S| \ge 0
\label{eq:h-submodularity}
\end{equation}
and so we have that $h$ is submodular.
We have to show that for any sets $S$, $T$, we have the inequality
$g(S) + g(T) \geq g(S \cap T) + g(S \cup T)$.
When $S \subseteq T$ or $T \subseteq S$, this inequality \rev{is immediate}; so suppose in the remainder
of the verification that $|S \setminus T| \geq 1$
and $|T \setminus S| \geq 1$.
By the submodularity of $h$,
it suffices to verify
$g^-(S) + g^-(T) \geq g^-(S \cap T) + g^-(S \cup T)$.
In the case that at least one of $S$, $T$ is not realized, assume for notation
that $T$ is not realized; we have $g^-(T) = g^-(S \cup T)$, so then the needed inequality
follows from the monotonicity of $g^-$.
The remaining case is where both $S$, $T$ are realized.
In this case, we have, with base $M$ logarithms,
\begin{align*}
g(S) + g(T)
& = \alpha \log |F_S| + \alpha \log |F_T| + h(S) + h(T) \\
& = \alpha \log |F_S| + \alpha \log |F_{S \cap T}| \avgdeg_C (T|S \cap T) + h(S) + h(T) \\
& \geq \alpha \log |F_S| + \alpha \log |F_{S \cap T}| + \alpha \log \frac{\maxdeg_C (T|S \cap T)}{M^{\epsilon}} + h(S) + h(T) \\
& =  \alpha \log |F_S| + \alpha \log |F_{S \cap T}| + \alpha \log \maxdeg_C (T|S \cap T) - \alpha \epsilon + h(S) + h(T) \\
& \geq \alpha \log |F_S| + \alpha \log |F_{S \cap T}| + \alpha \log \maxdeg_C (T|S \cap T) + h(S \cap T) + h(S \cup T) \\
& = g^-(S \cap T) + \alpha \log [|F_S|\maxdeg (T|S \cap T)] + h(S \cap T) + h(S \cup T).
\end{align*}
Here, the first inequality follows from $(M,\epsilon)$-uniformity
and the
second inequality follows from the above inequality~(\ref{eq:h-submodularity}) on $h$ along with the observation that $2 \epsilon - \alpha \epsilon \geq 0$.
In order to conclude the argument, it suffices to argue that
$\alpha \log [|F_S|\maxdeg (T|S \cap T)] \geq  g^-(S \cup T)$.
\begin{itemize}
\item
In the case that $S \cup T$ is realized, then
$|F_S|\maxdeg (T|S \cap T) \geq |F_{S \cup T}|$, and we are done.
\item
In the case that $S \cup T$ is not realized, we have that
$(S,T)$ is not a realizable pair, which implies that
$ \log [|F_S|\maxdeg (T|S \cap T)] > w^+$, and we are done.

\end{itemize}

\rev{We have argued that $g$ is monotone and submodular.
We verify the three claimed properties as follows.}

\rev{(1)} The configuration $C$ is assumed to be non-empty, so
$|F_\emptyset| = 1$, implying $g(\emptyset) = 0 + h(\emptyset) = 0$.

\rev{(2)} We show that
any set $S$ satisfying $g(S) \leq w$ is realized.
For any set $S$ satisfying this inequality,
we have $g(S) \leq w = w^+ / (1+\delta) \leq (1-\epsilon^{1/4})w^+ < (1-\epsilon^{1/3})w^+ \leq \alpha w^+ + h(S)$.  But then by definition of $g$,
we have that $g^-(S) < \alpha w^+$ and
thus $S$ is realized.

\rev{(3)}
\rev{Suppose that $S$ is realized, that is, suppose $S \in \R$, and that
$|F_S| \leq M$; we show $g(S) \leq 1$.}
Since $|F_S| \leq M$, we have
\rev{$g^-(S) \leq \alpha = 1 - \epsilon^{1/3}$.}
By our initial observation on $h$, we have
$h(S) \leq \epsilon^{1/3}$, and so
$g(S) \leq (1 - \epsilon^{1/3}) + \epsilon^{1/3} = 1$.
\end{proof}

\begin{lemma} (correctness)
\label{lemma:guanaco-correct}
On a configuration $C_0$ with hypergraph $H$,
a call to $\guanaco_{H,X,\delta}(C_0)$ returns
$\ans(C_0) \upharpoonright X$.
\end{lemma}

\begin{proof}
Consider the recursion tree of
$\guanacorec(C_0)$ when $\guanaco$ is run on a configuration $C_0$.
We claim that each call $\guanacorec(C)$
returns a set $A$ where
\[(\ans(C) \res X) \subseteq A \subseteq (\ans(C_0) \res X).\]
This suffices, since then the call $\guanacorec(C_0)$
returns the set $A = \ans(C_0) \res X$.

First, observe that each configuration resulting from a call to $\cleanup$ is a tightening of $C$, by
Theorem~\ref{thm:cleanup}, and also observe that realizing a pair also tightens a configuration.  Thus, every configuration
encountered in a call $\guanacorec(C)$
is a tightening of $C_0$.

We prove the claim by induction on the structure of the recursion tree.
Consider a single execution of $\guanacorec$.
Let us denote $m = |\C'|$;
 with respect to the $i$th iteration of the for loop, we use
  $C'_i$ to denote the configuration from $\C'$
 considered and $A_i$ to denote the set of
 answers added to $A$.
\begin{itemize}
\item
Consider the base case, where the execution of $\guanacorec(C)$
makes no recursive calls.
Each time line~\ref{alg:gr-find-td} is encountered,
the configuration $C'_i$ satisfies the assumptions of
Lemma~\ref{lemma:submodular}.
\rev{Thus by Lemma~\ref{lemma:submodular}, there exists a monotone submodular function $g$ satisfying properties (1), (2), and (3) from that lemma's statement.
By (1), we obtain that $g$ is a polymatroid, and by (3), we have that $g$ satisfies edge domination.  Thus $g$ is in the set $S(H)$ defined in Section~\ref{subsec:submodular-width},
and since $w = \subw(H,X)$, by the definition of submodular width
there exists a free-connex tree decomposition $(T,\chi)$ of $H,X$
where $\max_{t \in V(T)} g(\chi(t)) \leq w$.
It follows from (2) that each set $\chi(t)$ (over $t \in V(T)$)
is realized with respect to $C'_i$.
We obtain that $(T,\chi)$ is a free-connex tree decomposition of $H,X$
where all bags are realized with respect to $C'_i$, and so
line~\ref{alg:gr-find-td} of $\guanacorec$ successfully finds a tree decomposition.}

Thus Yannakakis is run on each configuration $C'_i$, resulting in sets
$A_i$ where $(\ans(C'_i) \res X) \subseteq A_i$
(recall Note~\ref{note:yannakakis}).
We then have $\bigcup_i (\ans(C'_i) \res X) \subseteq \bigcup_i A_i$, and since $\C'$ represents~$C$, we have
$(\ans(C) \res X) \subseteq \bigcup_i A_i$.
Moreover each $A_i$ is a subset of $\ans(C_0) \res X$
since~$A_i$ is obtained by using a free-connex tree decomposition of $H, X$ which covers all edges of $H$.
Thus $A = \bigcup_i A_i$ satisfies the stated condition.

 \item
Consider the inductive case, where the execution of $\guanacorec(C)$
makes recursive calls.
In the case that a recursive call is made after examining $C'_i$,
 let
$C''_i$
denote the configuration from line~\ref{alg:gr-realize}.
In this case,  we have
$(\ans(C''_i) \res X) \subseteq A_i \subseteq (\ans(C_0) \res X)$ by induction;
we also have $\ans(C''_i) = \ans(C'_i)$.  In the case that no recursive call is made after examining~$C'_i$, by reasoning as in the base case, we have
$(\ans(C'_i) \res X) \subseteq A_i \subseteq (\ans(C_0) \res X)$.
Thus, for all $i$ we have
\begin{center}
$(\ans(C'_i) \res X) \subseteq A_i \subseteq (\ans(C_0) \res X)$,
\end{center}
so as $\C$ represents~$C$,
we obtain
\begin{center}
$(\ans(C) \res X) = \bigcup_i (\ans(C'_i) \res X) \subseteq \bigcup_i A_i \subseteq (\ans(C_0) \res X)$.
\end{center}
Thus $A = \bigcup_i A_i$ satisfies the stated condition.
\end{itemize}
\end{proof}

Collecting together Lemmas~\ref{lemma:guanaco-termination} and~\ref{lemma:guanaco-correct}, we obtain our main result.

\begin{theorem}
Fix a hypergraph $H$ with vertex set $V$, fix
$X \subseteq V$, and fix $\delta > 0$.
On any input configuration $C_0$ whose hypergraph is $H$,
the call $\guanaco_{H,X,\delta}(C_0)$
returns $\ans(C_0) \res X$
within time
$O(N^{\subw(H,X) + \delta} + \out)$.
\end{theorem}

\begin{proof}
The call runs within the stated time due to Lemma~\ref{lemma:guanaco-termination},
and returns $\ans(C_0) \res X$ due to
Lemma~\ref{lemma:guanaco-correct}.
\end{proof}

\begin{note}
In the Boolean case, that is, the case where $X = \emptyset$,
the algorithms need only return \emph{yes} or \emph{no}
depending on whether or not there exists an answer to
the input configuration $C_0$, or, equivalently, whether or not
the empty assignment is in $\ans(C_0) \res X$ or not
(respectively).

In this case, the algorithms can be simplified as follows.
In $\guanaco$, in line~\ref{alg:guanaco-emptyset}, a \emph{no} can simply be returned.
In $\guanacorec$, adding results to $A$ via taking a union is tantamount to simply taking the boolean OR of the decisions returned by recursive calls and by the Yannakakis algorithm:
the only thing that would possibly be added to $A$ would be the empty assignment.
  So,
in line~\ref{alg:guanacorec-recurse}, if the recursive call returns \emph{yes}, then \emph{yes} can be returned immediately;
line~\ref{alg:gr-return-at-end} can be replaced with returning \emph{no};
and
lines~\ref{alg:gr-find-td} and~\ref{alg:gr-yannakakis}
can in fact be replaced with returning \emph{yes}: it was shown
in the proof of Lemma~\ref{lemma:guanaco-correct}
that a tree decomposition of the stated form (that is, a tree decomposition of~$H$ where all bags are realized) always exists, and in this case, as $C'$ is $\pi$-consistent and nonempty, by the known dynamic programming algorithm on tree decompositions, there is an answer to $C'$ restricted to relations defined on bags of $H$, which implies that there is an answer to the original configuration~$C_0$.
\end{note}

\rev{
\section{Conclusion}
\label{sec:conclusion}
Every conjunctive query can be evaluated in polynomial time with a certain exponent.
To date, the smallest known value of this exponent that can be achieved using combinatorial algorithms is the submodular width of the query.
However, algorithms achieving this exponent are often complex and rely on information theory, LP duality, and sub-probability measures---among other tools~\cite{DBLP:conf/pods/Khamis0S17,theoretics:13722,panda-express,10.1145/3788852.3805595}.
In contrast to these algorithms, we presented in this paper a significantly simpler algorithm, called $\guanaco$, that can evaluate any conjunctive query in polynomial time with an exponent equal to the submodular width plus an arbitrarily small $\delta$.
The algorithm and its analysis are natural and do not rely on any of the aforementioned tools.

We conclude with some open problems.}

\paragraph*{\rev{Removing the extra $\delta>0$ in the exponent.}}
\rev{
    It is not clear whether the extra $\delta$ in the exponent can be removed.
    This delta also seems to manifest in the \emph{Jaguar} algorithm~\cite{jaguar},
    and larger slacks in the exponent also appear in prior algorithms~\cite{DBLP:journals/jacm/Marx13,berkholz_et_al:LIPIcs.MFCS.2019.58}.
    This non-tightness in the exponent seems to be a common weakness of \emph{primal algorithms}, by which we mean algorithms that work directly in the space of polymatroids.
    The only known algorithms that bypass this weakness are \emph{dual algorithms}, by which we mean algorithms that operate in the space of Shannon inequalities---such as \emph{PANDA} and \emph{PANDAExpress}~\cite{DBLP:conf/pods/Khamis0S17, theoretics:13722, panda-express}.}
\paragraph*{\rev{Handling arbitrary degree constraints.}}
\rev{
    While we view $\guanaco$ as being significantly simpler than prior algorithms for conjunctive query evaluation such as \emph{PANDA}~\cite{DBLP:conf/pods/Khamis0S17, theoretics:13722, panda-express} and \emph{Jaguar}~\cite{jaguar}, the downside is that we have only shown it to achieve a complexity in terms of the input database size~$N$; unlike \emph{PANDA} and \emph{Jaguar}, it has not been shown to handle arbitrary {\em degree constraints} on the input relations~\cite{DBLP:conf/pods/Khamis0S17}.}


\begin{acks}
H. Chen acknowledges the support of EPSRC grants EP/V039318/1 and EP/X03190X/1.
\end{acks}

\bibliographystyle{ACM-Reference-Format}
\bibliography{bib,hubiebib}

\end{document}